\documentclass[english,a4paper,12pt]{article}
\usepackage[utf8]{inputenc}
\usepackage[T1]{fontenc}
\usepackage{titling}
\usepackage{hyperref}
\usepackage{amsmath, amssymb}
\usepackage{pgf}
\usepackage{amsfonts}
\usepackage{dsfont}
\usepackage{mathrsfs}
\usepackage{mathabx}
\usepackage{stmaryrd}
\usepackage{makeidx}
\usepackage{amsbsy}
\usepackage{amsthm}
\usepackage{color}
\usepackage{fullpage}
\usepackage{enumitem}
\usepackage[vcentermath]{youngtab}
\usepackage{marginnote} 
\usepackage{mdframed}	
\newmdenv[topline=false, bottomline=false, skipabove=\topsep, skipbelow=\topsep]{siderules}
\usepackage{tikz}
\usetikzlibrary{decorations.text,calc,arrows.meta}
\usepackage[vcentermath]{youngtab}
\usepackage{tkz-euclide}
\usetkzobj{all}
\usepackage{pgfplots}
\usetikzlibrary{shapes}
\usetikzlibrary{decorations.shapes}
\usetikzlibrary{positioning}
\newtheorem{theorem}{Theorem}

\newtheorem{proposition}{Proposition}
\newtheorem{lemma}{Lemma}
\newtheorem{remark}{Remark}

\def\bM{{\mathbb M}}

\def\gA{{\mathfrak A}}

\def\gM{{\mathfrak M}}

\newcommand{\A}{\mathfrak{A}}

\newcommand{\ca}[1]{{\cal #1}}
\newcommand{\ben}{\begin{equation}}
\newcommand{\een}{\end{equation}}
\def\bena{\begin{eqnarray}}
\def\eena{\end{eqnarray}}

\def\cP{{\ca P}}

\renewcommand{\H}{\mathcal{H}}

\def\id{{\mathrm{id}}}
\def\1{{\mathds{1}}}

\newcommand{\dd}{{\rm d}}
\newcommand{\tr}{\operatorname{Tr}}

\renewcommand{\log}{\operatorname{ln}}
\renewcommand{\Re}{\operatorname{Re}}
\renewcommand{\Im}{\operatorname{Im}}

\renewcommand{\epsilon}{\varepsilon}

\def\dom{{\mathrm{dom}}}

\newcommand{\RR}{\mathbb{R}}
\newcommand{\CC}{\mathbb{C}}

\begin{document}
\title{Relative entropy close to the edge}

	\author{
	Stefan Hollands$^{1}$\thanks{\tt stefan.hollands@uni-leipzig.de} 
		\\ \\
{\it ${}^{1}$Institut f\" ur Theoretische Physik,
Universit\" at Leipzig, }\\
{\it Br\" uderstrasse 16, D-04103 Leipzig, Germany} \\
	}
\date{\today}
	
\maketitle
	
\begin{abstract}
We show that the relative entropy between the reduced density matrix of the vacuum state in some region $A$  and that of an excited state created by a unitary operator localized at a small distance $\ell$ of a boundary point $p$ is insensitive to the global shape of $A$, up to a small correction. 
This correction tends to zero as $\ell/R$ tends to zero, where $R$ is a measure of the curvature of $\partial A$ at $p$, but at a rate necessarily slower than $\sim \sqrt{\ell/R}$ (in any dimension). Our arguments are mathematically rigorous and only use model-independent, basic assumptions about quantum field theory such as locality and Poincare invariance.
\end{abstract}
\bigskip
\noindent {\small Keywords: entanglement, relative entropy, axiomatic quantum field theory, operator algebras, excited states.\
}
\section{Introduction}

The entanglement between a localized subsystem and its environment in a given quantum state is by now a very well 
investigated subject in quantum field theory (QFT). A basic physical picture which has been confirmed in many examples -- and which is supported also by 
certain formal arguments -- is that the dominant contribution to the entanglement arises from the strong correlations between 
degrees of freedom localized on either side and in the proximity of the surface separating the subsystem from the environment. These correlations are so 
strong, in fact, that quantities like the entanglement entropy diverge in typical states (such as the vacuum) in QFT. 

If this quantity is computed with some short distance cutoff, then in many cases, formal arguments show that the leading contributions are organized in
a series in the inverse cutoff, the dominant terms of which are related to local curvature invariants of the entangling surface, see e.g. \cite{solo} (replica trick) or 
\cite{rang} (holographic methods) or \cite{sanders_2} (operator algebraic methods) and the many references therein. However, to our knowledge no universal, rigorous 
argument based just on the fundamental principles of QFT has been given to support the idea that the dominant contribution to entanglement (in ``typical'' states)
arises from local correlations across the entangling surface.   

In this paper, we provide such a universal argument that is based just on the standard features of locality (Einstein causality) and Poincare invariance in QFT.
Rather than proving asymptotic curvature expansions in a cutoff of the type described, our idea is to directly probe the ``dominant contributions'' of the reduced density matrix of the system near the boundary in an operational way. It is in more detail as follows. 

Let us say that the state of the QFT  is $|0 \rangle$, which we will take to be the vacuum for simplicity. Let $A$ be the spatial region of our subsystem, and $B$ its complement. The reduced density matrix is then given (formally) by $\rho_A = \tr_B  |0 \rangle \langle 0 |$. 
We want to ask how this reduced density matrix looks like from the point of view of observables in $A$ localized very near a point, $p$, on the boundary of $A$, and we would like to make a quantitative statement to the effect that the dominant part of $\rho_A$ with respect to such observables does not change if 
we deform $A$ sufficiently far away from $p$.  

More precisely, we consider two regions (systems) $A_1, A_2$ whose boundaries coincide near $p$ but may differ further away from $p$. The corresponding reduced density matrices are called $\rho_j = \rho_{A_j}, j=1,2$. Consider now a unitary operator $U$ that is localized in a very small ball of size $\ell$ within $A_1$ and $A_2$ at the, very short, distance $\ell$ away from the point $p$ where $A_1$ and $A_2$ touch, see fig. \ref{fig:regions}. 

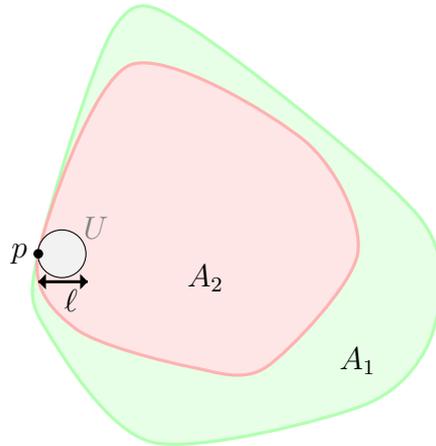
\begin{figure}[h!]
\begin{center}
\begin{tikzpicture}
\filldraw [color=green!30, fill=green!10, very thick] plot [smooth cycle] coordinates {(0,0) (1,3) (2,3) (5,.5) (5.2,-1) (4.3,-2) (1.5,-2.5) (.1,-1)};
\filldraw [color=red!30, fill=red!10, very thick] plot [smooth cycle] coordinates {(0,0) (1.2,2.5) (3.5,1.5) (4.2,0) (3,-1.5) (2,-1.5) (.5,-1)};
\filldraw[color=black, fill=gray!10] (.31,0) circle (9pt);
\draw[line width=1pt, >={Triangle[length=1mm,width=2mm]}, ->] (0,-0.37) -- (.65,-0.37);
\draw[line width=1pt, >={Triangle[length=1mm,width=2mm]}, ->] (.65,-0.37) -- (0,-0.37);
\node[below] at (2.2, 0) {$A_2$};
	\node[below] at (4.2, -1.1) {$A_1$};
	\node[right] at (2, 0.5) {$$};
	\node[right] at (.2, -.62) {$\ell$};
	\node[right] at (.45, .34) {\textcolor{gray}{$U$}};
	\node[left] at (.0, 0) {$p$};
	\filldraw (-.0,0) circle (1.6pt);
\end{tikzpicture}
\end{center}
\begin{center}
    \caption{The two regions $A_1$ and $A_2$. The gray blob indicates the localization of $U$.}
\label{fig:regions}
\end{center}
\end{figure}

 We would like to say that with respect to all such unitary operators, the reduced density matrices $\rho_1$ and $\rho_2$ look alike. To turn this into a quantitative statement, we look at the excited state $U |0\rangle$, with reduced density matrices given, obviously, by $U \rho_j U^*$ with respect to the $A_j$. The ``distance'' in state space between 
$\rho_j$ and $U \rho_j U^*$ should in this setup be nearly the same for $j=1$ or $=2$ (i.e. for $A_1$ and $A_2$) if indeed the reduced density matrices are insensitive to changes 
of the region far away from the point $p$ near which $U$ is localized. 

In this paper, we will use the relative entropy $S(\rho / \sigma)$ between two states as the natural distance measure\footnote{The v. Neumann entropy of $\rho_j$ etc. is not well-defined in QFT since the algebras of observables are of type $III$, see e.g. \cite{sanders_2,Witten:2018zxz} for a discussion of this well-known fact. By contrast, the relative entropy is defined for any type and should therefore regarded as the primary entropy concept in QFT.}, and we will show rigorously that 
\ben
\label{Sineq1}
S(\rho_1 / U \rho_1 U^*) -  S(\rho_2 / U \rho_2 U^*) 
=  O\bigg(
 \exp \left[- (\log \sqrt{R/\ell})^\alpha \right] \bigg)
 \quad \text{ as $\ell \to 0$.}
\een
Here, $R$ characterizes the curvature of the boundary of the 
smaller region at the point $p$ and $0< \alpha < 1$ is a parameter characterizing how much energy is created by $U$ from the vacuum: 
If $U^* |0\rangle$ is decomposed in an energy eigenbasis $|E\rangle$ with respect to the generator of boosts associated with the half-space 
touching $p$, 
then we ask $|\langle E | U^* |0\rangle |^2 \le O(e^{-|E|^\alpha})$. As we shall see, due to the sharp localization of $U$ such a behavior is possible in general 
only for $\alpha < 1$, but not for $\alpha=1$, which can be seen as a manifestation of the Heisenberg uncertainty principle. Thus, convergence 
in \eqref{Sineq1} necessarily falls slightly short of $O(\sqrt{\ell/R})$, which would correspond to $\alpha=1$.

To prove our result, we use operator algebraic methods, in particular methods from Tomita-Takesaki modular theory as well as Araki's definition of the 
relative entropy for general v. Neumann algebras. The basic principle is that the modular flow associated with $A_1$
when applied to $U$ will stay within the local algebra associated with $A_2$ for longer and longer as $\ell \to 0$; in fact, the maximum flow time up to which this is the case goes like $\sim (2\pi)^{-1} \log(R/\ell)$. As already shown in a classic paper by Fredenhagen \cite{fredenhagen_5}, this gives some control over the modular operators associated with $A_1$ and $A_2$. We improve and extend these methods so as to be able to obtain the bound \eqref{Sineq1}, the precise formulation of which is provided in thm. \ref{ttouch} below.

\medskip
\noindent
{\bf Notations and conventions:} Our use of the big-O-notation is the following. We say that $f(x) \le (\ge) O(g(x))$ as $x \to \infty$ if 
there is an $x_0$ and positive constants $C$ (resp. $c$) such that $f(x) \le Cg(x)$ (resp. $f(x) \ge cg(x)$) for $x \ge x_0$. When both relations hold, we 
say $f(x) = O(g(x))$. 

\section{Relative entropy between vacuum and an excited state}\label{sect1}

We first recall the definition of the relative entropy in terms of modular operators due to Araki \cite{araki_3} and then state our main technical result. It will be proven in sec. \ref{sect2} and then used in sec. \ref{sect3} to demonstrate \eqref{Sineq1}, see thm. \ref{ttouch}. For details on 
operator algebras in general we refer to \cite{Bratteli} and for a recent survey of operator algebraic methods in quantum information theory in
QFT, we refer to \cite{sanders_2}. A nice exposition directed towards theoretical physics audience is \cite{Witten:2018zxz}.

Let $\gM$ be a v. Neumann algebra\footnote{An algebra of bounded operators that is 
closed in the topology induced by the size of matrix elements.} of 
operators on a Hilbert space\footnote{We always assume that $\H$ is separable.}, $\H$. We assume that $\H$ contains a ``cyclic and separating'' vector for $\gM$, that is, a unit vector $|\Omega \rangle$
such that the set consisting of $a|\Omega\rangle$, $a \in \gM$ is a dense subspace of $\H$, and such that $a|\Omega\rangle=0$ always implies $a=0$
for any $a \in \gM$. We say in this case that $\gM$ is in ``standard form'' with respect to the given vector. $\gM^+$ denotes the set of positive, self-adjoint elements in $\gM$ (which are always of the form $a=b^*b$ for some $b \in \gM$).

In this situation, one can define the Tomita operator $S$ on the domain ${\rm dom}(S) = \{ a|\Omega\rangle \mid a \in \gM\}$ by 
\ben
Sa|\Omega \rangle = a^* |\Omega \rangle
\een
The definition is consistent due to the cyclic and separating property. It is known that $S$ is a closable operator, and we denote its closure by the same symbol. This closure has a polar decomposition denoted by $S=J\Delta^\frac12$, with $J$ anti-linear and unitary and $\Delta$ self-adjoint and non-negative. 
Tomita-Takesaki theory concerns the properties of the operators $\Delta, J$. The basic results of the theory are the following, see e.g. \cite{Bratteli}: 
\begin{enumerate}
\item 
$J \gM J = \gM'$, where the prime denotes the {\bf commutant} (the set of all bounded operators on $\H$ 
commuting with all operators in $\gM$) and $J^2 = 1, J\Delta J = \Delta^{-1}$, 

\item 
If $\sigma^t(a) =  \Delta^{it} a \Delta^{-it}$, then $\sigma^t \gM = \gM$ and $\sigma^t \gM' = \gM'$ for all $t \in \RR$,

\item
The positive, normalized (meaning $\omega(a) \ge 0 \, \, \forall a \in \gM^+, \omega(1) = 1$) 
linear {\bf expectation functional} 
\ben 
\omega(a) = \langle \Omega |a\Omega \rangle
\een
satisfies the {\bf KMS-condition} 
relative to $\sigma^t$. This condition states that for all $a,b \in \gM$, the bounded function 
\ben
\label{Fdef}
t \mapsto F_{a,b}(t) = \omega (a \sigma^t(b)) \equiv \langle \Omega | a \Delta^{it} b \Omega \rangle
\een
has an analytic continuation to the strip $\{ z \in \CC \mid -1 < \Im z < 0 \}$ with the property that its boundary value for $\Im z \to -1^+$ exists and
is equal to
\ben
\label{KMS}
F_{a,b}(t-i) = \omega (\sigma^t(b) a).
\een

\item 
Any normal (i.e. continuous in the weak$^*$-topology) positive linear functional $\omega'$ on $\gM$ has a unique vector representative 
$|\Omega' \rangle$ in the natural cone 
\ben
\cP^\sharp = \overline{\{ \Delta^{1/4} a |\Omega \rangle \mid a \in \gM^+\}}=
\overline{\{ aj(a) |\Omega \rangle \mid a \in \gM\}}, 
\een
where the overbar means closure and $j(a) = JaJ$. The state functional is thus
$\omega'(a) = \langle \Omega' |a\Omega' \rangle$ for all $a \in \gM$. 
\end{enumerate}

\medskip
\noindent
{\bf Key example}: These claims are easy to verify in the ``type $I_n$'' case $\gM = M_n(\CC) \otimes 1_n$, which acts 
on the first tensor factor in the Hilbert space $\H = 
\CC^n \otimes \CC^n$ (note that elements in the Hilbert space can be identified with matrices on which $\gM \cong M_n(\CC)$ acts by left multiplication). 
The commutant is $\gM' = 1_n \otimes M_n(\CC)$. A vector $|\Omega \rangle$ in this Hilbert space 
is cyclic and separating if $|\Omega \rangle = \sum_{j=1}^n \sqrt{p_j} |j \rangle \otimes |j\rangle$ in some ON basis $|j\rangle$ and iff all $p_j > 0$, $\sum_{j=1}^n p_j=1$. In the example, the corresponding state functional can be represented by the reduced density matrix 
\ben
\rho = \sum_{j=1}^n p_j |j \rangle \langle j|, \quad \omega(a) = \tr_{\CC^n}(a \rho) \quad (a \in \gM \cong M_n(\CC)), 
\een
and we see that the cyclic and separating property in a sense says that this reduced density matrix is ``as mixed as possible''.
The modular operator is given by 
\ben
\Delta^\frac12 = \rho^{\frac12} \otimes \rho^{-\frac12},  
\een
while the modular conjugation is given by $J(a \otimes 1_n) = 1_n \otimes \bar a$, where the overbar means the element-wise complex conjugation of a matrix $a$. 
The modular flow is, therefore, $\sigma^t(a) = \rho^{it} a \rho^{-it}$. The ``modular Hamiltonian'', 
\ben 
\log \Delta = \log \rho \otimes 1_n - 1_n \otimes \log \rho
\een 
can be split in the present example into a part belonging to $\gM$ (the first term) and a part belonging to $\gM'$ (the second term). This split is impossible for general v. Neumann algebras, in particular for the type $III_1$-factors appearing in quantum field theories\footnote{The possibility of making the split implies that 
$\sigma^t$ is inner, i.e. can be written as $\sigma^t(a) = u(t) a u(t)^*$ for unitaries $u(t)$ in $\gM$. One characterization of type $III$ v. Neumann algebras 
is that $\sigma^t$ precisely cannot be inner for any normal state $\omega$.}.  The natural cone consists of the self-adjoint, 
positive semi-definite matrices in $\H$.

\medskip
\noindent

A generalization of this construction is that of the relative modular operator, flow etc. \cite{araki_1}. For this purpose, let $\omega'$ be a normal state on $\gM$, $|\Omega' \rangle$ its unique vector representative in the natural cone in $\H$, which is assumed (for simplicity) to be cyclic and separating, too. 
Then we can consistently define
\ben
S_{\omega, \omega'} a |\Omega' \rangle = a^* |\Omega \rangle
\een
form the closure, and make the polar decomposition $S_{\omega, \omega'}=J_\omega^{} \Delta^\frac12_{\omega,\omega'}$. 
The {\bf relative entropy} is defined by 
\ben
S(\omega / \omega') = \langle \Omega | (\log  \Delta_{\omega,\omega'}) \Omega \rangle . 
\een
In the above example, we get $\Delta^\frac12_{\omega, \omega'} = \rho^{\frac12} \otimes \rho^{\prime -\frac12}$ and thus 
$S(\omega / \omega') = \tr \rho ( \log \rho - \log \rho')$, where $\rho'$ is the density matrix associated with $\omega'$. 
The relative entropy has many beautiful properties, the important ones of which were already derived by Araki \cite{araki_3}.  It is e.g. never negative, but can be infinite, is decreasing under completely positive maps, is jointly convex in both arguments, etc. The physical interpretation of $S(\omega / \omega')$ is the amount of information gained if we update our belief about the system from the state $\omega$ to $\omega'$.

In this paper, we are interested in the special case when $\omega'=\omega_U$, where
\ben
\omega_U(a) \equiv \omega(U^* a U) = \langle U \Omega | a U\Omega \rangle,  
\een
and where $U$ is a unitary operator from $\gM$. In applications, $\omega$ is for instance the vacuum state and $\omega_U$ represents an excited state. 
The corresponding vector representative in the natural cone is $|\Omega_U \rangle=Uj_\omega (U)|\Omega \rangle$, with $j_\omega (a)= J_\omega a J_\omega$. 
Going through the definitions, one finds immediately that $j_\omega(U) \Delta^{1/2}_\omega j_\omega(U^*) = \Delta^{1/2}_{\omega, \omega'}$, implying that
\ben\label{Srel}
S(\omega/  \omega_U) = -\langle U^* \Omega |  (\log \Delta) U^*\Omega \rangle, 
\een
where $\Delta$ is the modular operator of the original state $\omega$. More specifically, we are in the following setup:

\medskip
\noindent
{\bf Basic setup:} We assume that we have an inclusion $\gM_1 \supset \gM_2$ of v. Neumann algebras in standard form on a Hilbert space $\H$ with cyclic and separating vector $|\Omega\rangle$ (for both $\gM_j$) The associated modular operators are called
$\Delta_1, \Delta_2$ and the modular flows are called $\sigma^t_j(a) = \Delta_j^{it} a \Delta_j^{-it}, j=1,2$. Note that if $a \in \gM_2$, then 
$\sigma^t_1(a)$, the modular flow of $\gM_1$, will leave $\gM_2$ for $t \neq 0$ in general.

\medskip
\noindent

Given a unitary $U \in \gM_2$, we can then define the relative entropy between $\omega$ and $\omega_U$ with respect to $\gM_1$ (i.e. the states are viewed as functionals on $\gM_1$) or with respect to $\gM_2$ (i.e. the states are viewed as functionals on $\gM_2$). These relative entropies are denoted by 
\ben
S_j(\omega / \omega_U) := S(\omega \restriction_{\gM_j} / \omega_U \restriction_{\gM_j}), \quad j=1,2
\een
 and are in general different. (The monotonicity of the relative entropy \cite{araki_5} gives $S_1 \ge S_2$.) Our main technical result is:

\begin{theorem}\label{thm1}
Let $U \in \gM_2$ be a unitary such that $\sigma_1^t(U) \in \gM_2$ for $|t| \le \tau$. 
\begin{enumerate}
\item For $n>1$ we have that
\ben
| S_1(\omega / \omega_U) -  S_2(\omega / \omega_U)| 
\le  O(\tau^{-n+1}) \omega_{U^*}((1+H_-)^n) , 
\een
for large $\tau$ uniformly in $U$. 
\item For $0<\alpha<1$ we have that
\ben
| S_1(\omega / \omega_U) -  S_2(\omega / \omega_U) | 
\le  O(\tau^{1-\alpha}e^{-(\pi\tau)^\alpha}) \, \omega_{U^*}(e^{H_-^\alpha}) , 
\een
for large $\tau$ uniformly in $U$. 
\end{enumerate}
Here, $H_-= -E_1^-\log \Delta_1 \ge 0$ is the negative part of the modular Hamiltonian, where $E^-_1$ is the spectral projector of $\log \Delta_1$ (see \eqref{decomp}) associated with the negative part of the spectrum.
\end{theorem}

This theorem is a direct consequence of \eqref{Srel} and prop. \ref{prop1}. It expresses that the difference between the relative entropies goes to zero for 
unitaries having a large $\tau$, i.e. unitaries staying inside $\gM_2$ for long under the modular flow of $\gM_1$. 
The rate of decay depends on the property of the unitary. Roughly speaking, the less energetic the excited state 
$\omega_{U^*}$ is with respect to the negative part $H_-=-E^-_1 \log \Delta_1$ of the modular Hamiltonian, the faster the decay of the difference as $\tau$
goes to infinity. 
For an exponential decay (i.e. $\alpha = 1$ in the second case), one would need a non-trivial unitary $U \in \gM_2$ such that the vector $U^* |\Omega \rangle$ is in the domain of the {\em inverse} modular operator $\Delta_1^{-1}$. Such unitaries typically do not exist, see sec. \ref{sect3} for an illustrative example. In some sense this is a consequence of the uncertainty principle. Thus, we need to content ourselves with a sub-exponential decay. 

We will illustrate the meaning of this result in the context of relativistic quantum field theory in sec. \ref{sect3}. 

\section{Technical results}
\label{sect2}

 In this section, we assume the Basic Setup described above. Our first lemma is the following:

\begin{lemma}\label{lemma1}
Let $\epsilon(k)$ be any non-negative, non-increasing continuous function such that 
$\int^\infty \epsilon(k)/k \, \dd k < \infty$. 
Then there exists a positive real valued continuous 
function $g(k)=O(e^{-|k| \epsilon(|k|)})$ (as $|k| \to \infty$) such that for all 
$a \in \gM_2$ having the property $\sigma_1^t(a) \in \gM_2$ for $|t| \le \tau>1$: 
\ben
\label{main}
\begin{split}
0 &\le \langle a\Omega |  (1+e^u \Delta_1)^{-1} a \Omega\rangle  - \langle a\Omega | (1+e^u \Delta_2)^{-1} a \Omega\rangle \\
&\le e^{-2\pi \gamma \tau} R_a(u)^{1-\gamma} 
\, 
\left\{ 
\left\| g(\log \Delta_1 +u) a \Omega \right\|^2
 + e^u
\left\| g(\log \Delta_1 -u) a^* \Omega\right\|^2 
 \right\}^\gamma 
 \end{split}
\een
for all $u \in \RR, 0 \le \gamma \le 1$, where 
\ben\label{Fdef}
R_a(u) =  
\begin{cases}
\langle a\Omega |  (1+e^u \Delta_1)^{-1} a \Omega\rangle & \text{for $u>0$,}\\
e^{u} \langle a^* \Omega | (1+e^{-u} \Delta_1)^{-1} a^* \Omega \rangle & \text{for $u\le 0$.}  
\end{cases}
\een
\end{lemma}

\begin{remark}
Below we will need the lemma only with $g = 1$. If we also set $\gamma=1$, the statement is already proven in \cite{fredenhagen_5}. 
The case of general $g$ can be interesting for other applications, e.g. if $u$ remains bounded or if $a^* |\Omega \rangle$ is very small, i.e. 
if $a$ is approximately a creation operator. 
\end{remark}

\begin{proof}
   We let $S_i$ be the Tomita operators for $\gM_i$ with polar decompositions $S_i = J_i \Delta_i^{1/2}$.  Note that, 
since $\gM_2 \subset \gM_1$, $\dom(S_2) \subset \dom(S_1)$. 
The set $\dom(S_1)$ is a Hilbert space called $\H_1$ with respect to the inner product (graph norm)
\ben
( \Phi , \Psi ) = \langle \Phi | \Psi \rangle + e^u \langle S_1 \Psi | S_1 \Phi \rangle = \langle \Phi | (1+e^u \Delta_1) \Psi \rangle . 
\een
Letting $I:\H_1 \to \dom(S_1)$ 
be the identification map, one shows that  $I^{-1} \dom(S_2)$ is a closed subspace $\H_2 \subset \H_1$ with associated orthogonal projection $P_2$. 
The operators $V_j= I^{-1}(1+ e^u\Delta_j)^{-1/2}$ are isometries from $\H$ to $\H_j$ ($j=1,2$) and their adjoints are $V_j^*= (1+ e^u\Delta_j)^{1/2} I P_j$
(with $P_1=1$). There follow the relations
\bena
I P_j I^* &=& I V_j^{} V_j^* I^* = (1 + e^u\Delta_j)^{-1}, \quad j=1,2\\
I^* &=& I^{-1}(1+ e^u\Delta_1)^{-1}, 
\eena
which can already be found in \cite{fredenhagen_5}. 

These relations imply that for all $a \in \gM_2$ 
\ben\label{frede}
\langle a\Omega |(1 + e^u\Delta_1)^{-1} a  \Omega \rangle -  \langle a\Omega | (1 + e^u\Delta_2)^{-1} a  \Omega \rangle
=  \| (1- P_2) I^{-1} (1 + e^u\Delta_1)^{-1} a \Omega \|^2 \ .
\een
The fact that the right side is manifestly non-negative already gives the left inequality in \eqref{main}. 
One way to estimate the right side is as follows. For $u>0$, we simply use that $\|1-P_2\|=1$ to get 
\ben
\label{step0}
\begin{split}
\| (1- P_2) I^{-1} (1 + e^u\Delta_1)^{-1} a \Omega \|^2 &\le \|  I^{-1} (1 + e^u\Delta_1)^{-1} a \Omega \|^{2} \\
&= \langle a\Omega |(1 + e^u\Delta_1)^{-1} a  \Omega \rangle,
\end{split}
\een
using  in the last step the definition of $I$. For $u \le 0$, we use that $(1- P_2) I^{-1} a |\Omega\rangle = 0$ since $a |\Omega \rangle$ is in the domain of $S_2$, meaning that $I^{-1}a |\Omega \rangle$ is in $\H_2$. 
Thus we can write
\ben
\label{step01}
\begin{split}
\| (1- P_2) I^{-1} (1 + e^u\Delta_1)^{-1} a \Omega \|^2 &= \|  (1- P_2) I^{-1} [(1 + e^u\Delta_1)^{-1} -1]a \Omega \|^{2} \\
&\le \|  I^{-1} [(1 + e^u\Delta_1)^{-1} -1]a \Omega \|^{2} \\
&= e^{u} \langle a \Omega | \Delta_1 (1+e^{-u} \Delta_1^{-1})^{-1} a \Omega \rangle\\
&= e^{u} \langle J_1 a^* \Omega |  (1+e^{-u} \Delta_1^{-1})^{-1} J_1 a^* \Omega \rangle\\
&= e^{u} \langle a^* \Omega | (1+e^{-u} \Delta_1)^{-1} a^* \Omega \rangle
\end{split}
\een
using again the definition of $I$ in the third step and $S_1 = J_1 \Delta^\frac12_1$ in the fourth line and $J_1 \Delta_1 J_1 = \Delta_1^{-1}$ in 
the last line. Together with \eqref{step0} this shows that 
\ben
\label{step2}
\| (1- P_2) I^{-1} (1 + e^u\Delta_1)^{-1} a \Omega \|^2 \le R_a(u), 
\een
which is our first way to estimate the right side of \eqref{frede}. 

A second way is as follows.
For a real number $y>0$, we write:
\ben\label{fourier}
\frac{1}{1+y} = \int_{-i\infty + 0}^{i\infty + 0} \Gamma(t) \Gamma(1-t) y^t \frac{\dd t}{2\pi i}
=\frac{i}{2} \int_\RR \frac{y^{it}}{\sinh[\pi (t+i0)]} \dd t  
\een
where the first equality is the standard Mellin-Barnes representation of the geometric series and the second follows from the properties of the Gamma function.
Therefore, by the spectral calculus
\ben
(1- P_2) I^{-1} (1 + e^u\Delta_1)^{-1} a |\Omega \rangle = \frac{i}{2} 
\int_\RR \frac{ \dd t \ e^{iut}}{\sinh[\pi (t+i0)]}  \, (1- P_2) I^{-1}  \Delta_1^{it} a |\Omega \rangle . 
\een
For $|t| < \tau$, we know by assumption that $\sigma_1^t(a)$ is in $\gM_2$, so
$\Delta_1^{it} a |\Omega \rangle=\sigma_1^t(a) |\Omega \rangle$ is in $\dom(S_2)$, so $I^{-1} \Delta_1^{it} a |\Omega \rangle$ is in $\H_2$, so $(1- P_2) I^{-1}  \Delta_1^{it} a |\Omega \rangle=0$. 
So we can effectively restrict the range in the integral to $|t| \ge \tau$ and drop the $i0$-prescription. A even better estimate is obtained if instead we choose a 
real-valued smooth function $\hat G(t)$ such that $\hat G(t) = 0$ for $t<-\frac12$ and $\hat G(t) = 1$ for $t>\frac12$ related to $G(k)$ via the Fourier transform, 
\ben
G(k) = \frac{1}{2\pi} \int_\RR e^{ikt} \hat G(t) \dd t . 
\een
Now define
\ben
f_\tau(y) = \Im \int_{0}^\infty \hat G\left(t-\tau+\frac12 \right) \frac{y^{-it}}{\sinh(\pi t)} \dd t. 
\een
It follows that    
\ben
\label{step1}
\begin{split}
\| (1- P_2) I^{-1} (1 + e^u\Delta_1)^{-1} a \Omega \|^2 &= \| (1- P_2) I^{-1} f_\tau(e^u\Delta_1) a \Omega \|^2\\
&\le \| I^{-1} f_\tau(e^u\Delta_1) a \Omega \|^2\\
&= \langle f_\tau(e^u\Delta_1) a \Omega | (1 + e^u\Delta_1) f_\tau(e^u\Delta_1) a  \Omega \rangle\\
&= \| f_\tau(e^u\Delta_1) a  \Omega \|^2 + e^u \| f_\tau(e^u\Delta_1) J_1 a^*  \Omega \|^2 \\
&= \| f_\tau(e^u\Delta_1) a  \Omega \|^2 + e^u \| f_\tau(e^u J_1\Delta_1 J_1)  a^*  \Omega \|^2 \\
&= \| f_\tau(e^u\Delta_1) a  \Omega \|^2 + e^u \| f_\tau(e^u \Delta_1^{-1})  a^*  \Omega \|^2 
\end{split} 
\een
using the definition of $I$ in third line and $S_1 = J_1 \Delta^\frac12_1$ in the fourth line and $J_1 \Delta_1 J_1 = \Delta_1^{-1}$ in 
the last line.

We claim that a $G(k)$ exists with the same fall-off for large $|k|$ as the function $g(k)$ stated in the lemma. More precisely, we state:
\begin{lemma}
There exists a smooth function $\hat G(t)$ such that $\hat G(t)=0$ for $t<-\frac12$, such that $\hat G(t)=1$ for $t>\frac12$, and such that
the (inverse) Fourier transform satisfies
\ben\label{bound}
|G(k)| \le c_0 \, e^{c_1 \Im k} e^{-|\Re k| \epsilon(|\Re k|)}
\een
for $k$ in the upper half plane, provided $|k|$ sufficiently large.
\end{lemma}
\begin{remark}
The proof shows that we can choose $c_1$ to be any constant $>\frac12$.
\end{remark}
\begin{proof}
Note that $G(k)$, if it exists, must be automatically analytic in the upper half plane $\Im k > 0$.  
That such functions exist is well-known. To prove the claimed fall off for imaginary $k$, we adapt a method by Ingham \cite{ingham}. We set 
$$
G(k) = -(2\pi i)^{-1}(k+i0)^{-1} \prod_{n=1}^\infty \frac{\sin \rho_n k}{\rho_n k}. 
$$
The product converges absolutely and uniformly in each finite domain of $k$ if the series of positive terms $\sum_n \rho_n$ is 
convergent. Furthermore, $G(k)$ is analytic in the upper half plane, and the Fourier transform of $k G(k)$ has support inside the interval $t \in [-\frac12,\frac12]$
provided $\sum_n \rho_n = \frac12$, essentially because the product of sinc functions in $k$-space corresponds to an infinite convolution of top hat functions 
in $t$-space, and the $n$-th convolution increases the support by an amount $\rho_n$ (see \cite{Bochner} for details). 
It follows that the Fourier transform $\hat G(t)$ of $G(k)$ is such that we have the desired properties 
$\hat G(t) = 0$ for $t<-\frac12$ and $\hat G(t) = 1$ for $t>\frac12$. 

Now it is compatible with choices already made to 
take $\rho_n$ non-increasing with $\rho_n \ge e\epsilon(n)/n$ for $n$ exceeding some $n_0$, and we set $\nu = \lfloor |\Re k| \epsilon (|\Re k|) \rfloor$. 
Furthermore, we let $N$ be the largest natural number such that $|k| \rho_n \ge 1$. Then we split the 
product defining $G(k)$ into factors in the range a) $1 \le n \le \nu$, b) $\nu < n \le N$, and c) $N<n$. 
For sufficiently large $|\Re k|$ we then have $|\Re k| \rho_\nu  \ge e$, and the factors in the range a) can 
consequently be estimated in absolute value by: 
\ben\label{bound0}
 \prod_{n=1}^\nu \frac{e^{\rho_n \Im k}}{\rho_n |\Re k|} 
 \le  \left( \prod_{n=1}^\nu e^{\Im k \rho_n} \right) \Bigg( \frac{1}{\rho_\nu |\Re k|} \Bigg)^\nu
 \le \exp \left( \Im k \sum_{n=1}^\nu \rho_n \right) e^{-|\Re k| \epsilon(|\Re k|)+1}. 
\een
The factors in the range b) can be estimated in absolute value by $e^{\Im k\sum_{n=\nu+1}^N \rho_n} \le e^{\Im k/2}$. The factors in the range c) 
can be estimated e.g. using the infinite product for the sinc function given in \cite{product} (putting $x=k\rho_n$, so $|x|<1$ in the range c)
\ben
\left| \frac{\sin x}{x} \right| = \prod_{j=1}^\infty \left| 1- \frac{x^2}{j^2 \pi^2} \right| = \prod_{j=1}^\infty \left(1 + \frac{|x|^4}{j^4 \pi^4}  - 2\frac{|x|^2}{j^2 \pi^2} \cos (2\arg(x)) \right)^{\frac12} < 1
\een
where the last step holds provided $2\cos (2\arg(k))=2\cos (2\arg(x)) > \pi^{-2}$. Thus, provided  $\delta |\Re k| > \Im k$ for some sufficiently small $\delta>0$, 
the modulus of the factors in c) is bounded by $1$. On the other hand, in the sector $\delta |\Re k| \le \Im k$, the modulus 
of the factors in c) is estimated simply by (putting $x=k\rho_n$) $\left| \frac{\sin x}{x} \right| \le e^{|x|}$, so in that sector the 
modulus of the product of the factors in c) is at most $e^{|k|
\sum_{n=N+1}^\infty \rho_n} \le e^{|k|/2}$. 

Multiplying our bounds 
for a),b),c) gives the claimed bound \eqref{bound} 
for $k$ in the upper half plane, provided $|k|$ sufficiently large, noting that in the sector $\delta |\Re k| \le \Im k$, this bound 
is compatible with an upper bound of the form $e^{c_2|k|}$ for $c_2>0$.
\end{proof}

We now use this knowledge to gain more information about the function $f_\tau(y)$. For convenience we write $y=e^k$.  
 Applying the convolution theorem, using the Fourier transform of $[\sinh(\pi(t+i0))]^{-1}$ (see
 \eqref{fourier}), 
 and the usual behavior of 
 the Fourier transform under a shift of the argument, the definition of $f_\tau(e^k)$ can be rewritten as
 \ben
 f_\tau(e^k) = \Re \int_{-\infty}^\infty G(p) e^{ip(\tau-\frac12)} (1+e^{k-p})^{-1} \dd p
 \een
Since $G(p)$ is analytic in the upper half plane, we can evaluate the integral by means of the residue theorem. 
The poles of $(1+e^{k-p})^{-1}$ in the upper half plane are at the points
$
p=k+2\pi i(n+\frac12), \quad n = 0, 1, 2, \dots . 
$
Since $G$ satisfies the decay condition \eqref{bound}, we can close the contour when $\tau>1$. 
Application of the residue theorem then gives
\ben
 f_\tau(e^k) = 2\pi \, \Im \left( e^{ik(\tau-\frac12)} e^{-\pi (\tau-\frac12)} \sum_{n=0}^\infty G(k+2\pi i(n+\tfrac{1}{2})) e^{-2\pi n(\tau-\frac12)} \right)
\een
which converges for $\tau>1$. 
Now we apply the bound \eqref{bound} to estimate the series term-by-term. This in combination with the geometric series gives 
\ben
 |f_\tau(e^k)| \le e^{-\pi \tau}  g(k)
\een
for a function $g(k)$ with the properties claimed in the lemma.
We insert this into the right side of \eqref{step1} and apply the functional calculus. Then we immediately get
\ben
\| (1- P_2) I^{-1} (1 + e^u\Delta_1)^{-1} a \Omega \|^2 \le e^{-2\pi \tau} (\left\| g(\log \Delta_1 +u) a \Omega \right\|^2
 + e^u
\left\| g(\log \Delta_1 -u) a^* \Omega\right\|^2).
\een
Combining this with \eqref{step2} in \eqref{frede} then gives the inequality claimed in the lemma.  
\end{proof}

It is clear that the term in curly brackets in \eqref{main} is bounded for instance by
\ben
\{ \dots \} \le c(\| a\Omega\|^2 + e^u \| a^*\Omega\|^2),
\een
choosing $g$ to be constant. 
Next we need a bound on $R_a$ as defined in eq. \eqref{Fdef}. For negative $u$, we trivially get the bound $R_a(u) \le e^u \| a^* \Omega\|^2$. 
For positive $u$, we decompose $R_a = R_a^+ + R_a^-$ with
\ben
R_a^+(u) = \int_0^\infty  (1+e^{k+u})^{-1} \langle a \Omega | E_1(\dd k) a \Omega \rangle \le (1+e^u)^{-1} \|E^+_1a \Omega\|^2. 
\een
 Here, $E^+_1$ is the spectral projection for the positive part of the spectrum in the decomposition 
of $\log \Delta_1$ given by
\ben
\label{decomp}
\log \Delta_1 = \int_\RR k E_1(\dd k) , \quad E^\pm_1 = \int_{\RR_\pm} E_1(\dd k).
\een 
Similarly, still for positive $u$, we choose some continuous function $\eta: \RR_- \to \RR_+$ and estimate
\ben
\begin{split}
R_a^-(u) &= \int_{-\infty}^0  (1+e^{k+u})^{-1} \langle a \Omega | E_1(\dd k) a \Omega \rangle \\
&\le 
\left(\inf_{k \ge 0}(1+e^{-k+u})\eta(-k) \right)^{-1} \,  \langle a \Omega | \eta(E_1^-\log \Delta_1) a \Omega \rangle .
\end{split}
\een
Here, $E^-_1$ projects onto the negative part of the spectral decomposition 
of $\log \Delta_1$. We now assume that $\eta$ is chosen in such a way that $u \mapsto [\inf_{k \ge 0}(1+e^{-k+u})\eta(-k)]^{-1}$ 
is an integrable function on $\RR_+$, which is the case e.g. if $\eta$ is bounded positively away from zero and if 
we choose $\eta(-k) = O(k^n)$ with $n>1$ or $\eta(-k) = O(e^{k^\alpha})$ with $\alpha>0$
when $k \to \infty$ (see also lemma \ref{l6}). 
Then \eqref{main} yields:

\begin{lemma}
Let $\eta: \RR_- \to \RR_+$ be a continuous function such that 
$ \RR_+ \owns u \mapsto [\inf_{k \ge 0}(1+e^{-k+u})\eta(-k)]^{-1}$ is integrable, and let 
$0 \le \gamma \le 1$. There is a constant $C$ independent of $u, \tau$ such that for all $u \le 0$:
\ben
\begin{split}
0 &\le \langle a\Omega |  (1+e^u \Delta_1)^{-1} a \Omega\rangle  - \langle a\Omega | (1+e^u \Delta_2)^{-1} a \Omega\rangle \\
& \le C \, e^{-2\pi \gamma \tau} \bigg(\| a\Omega\|^2 + e^u \| a^*\Omega\|^2 \bigg)^\gamma 
\bigg(
e^u \|a^* \Omega\|^2 
\bigg)^{1-\gamma} 
\end{split}
\een
whereas for all $u \ge 0$:
\ben
\begin{split}
0 &\le \langle a\Omega |  (1+e^u \Delta_1)^{-1} a \Omega\rangle  - \langle a\Omega | (1+e^u \Delta_2)^{-1} a \Omega\rangle \\
& \le C \, e^{-2\pi \gamma \tau} \bigg(\| a\Omega\|^2 + e^u \| a^*\Omega\|^2 \bigg)^\gamma
\bigg( \frac{\langle a \Omega |  \eta(E_1^-\log \Delta_1) a \Omega \rangle}{\inf_{k \ge 0}(1+e^{-k+u})\eta(-k)}  + 
\frac{\|E^+_1a \Omega\|^2}{1+e^u} \bigg)^{1-\gamma}
\end{split}
\een
for all $a \in \gM_2$ such that $\sigma_1^t(a) \in \gM_2$ for $|t| \le \tau$ 
with the property that $a|\Omega\rangle$ is in the domain of $\eta(E_1^-\log \Delta_1)$.
\end{lemma}
We now integrate the inequalities from this lemma against $u$ and use on the left side the operator identity
\ben
\log \Delta_2 - \log \Delta_1 =  
\int_{-\infty}^\infty  \left(
\frac{1}{1+ e^u \Delta_1}
-
\frac{1}{1+ e^u \Delta_2}
\right) \dd u \ ,
\een
where the integral is understood in the Cauchy principal value sense in the strong operator topology (it may not exist when applied to a vector not 
in the domain of both $\log \Delta_j$). The integration range is split into the following parts:
$u \in (-\infty, 0), [0, \pi\tau), [\pi\tau, \infty)$. For the first region, we take $\gamma = \frac12$, for the second region, we take $\gamma = 1$, 
and for the third region, we take $\gamma = 0$. To get a non-trivial bound, 
$\eta $ is chosen such that $u \mapsto [\inf_{k \ge 0}(1+e^{-k+u})\eta(-k)]^{-1}$ is an integrable function
and such that $a|\Omega\rangle$ is in the domain of $\eta(E_1^-\log \Delta_1)$. If we also take $a=U$ to be a unitary (implying that $\|U \Omega \| = 1 = 
\|U^* \Omega\|$), then we immediately obtain the following theorem. 

\begin{theorem}\label{th2}
Let $\eta: \RR_- \to \RR_+$ be a continuous function such that 
$\RR_+ \owns u \mapsto [\inf_{k \ge 0}(1+e^{-k+u})\eta(-k)]^{-1}$ is integrable.
There exists a constant $C$ not depending on 
$u,\tau$ such that
\ben
| \langle U\Omega | (\log \Delta_1) U \Omega \rangle - \langle U\Omega | (\log \Delta_2) U \Omega \rangle |
\le C \, K_\eta(\tau) (1+ \langle U \Omega |  \eta(E_1^-\log \Delta_1) U \Omega \rangle ), 
\een
for any unitary $U \in \gM_2$ such that $\sigma_1^t(U) \in \gM_2$ for $|t| \le \tau$, 
where 
\ben
K_\eta(\tau) =  \int_{\pi \tau}^\infty \left( \inf_{k \ge 0}(1+e^{-k+u})\eta(-k)\right)^{-1} \dd u .
\een
\end{theorem}

\begin{remark}
\label{improv}
Due to our assumption on $\eta$, $K_\eta(\tau)$ goes to zero as $\tau \to \infty$, but never faster than $e^{-\pi \tau}$. Variants of the above bound can
be obtained by taking the second integration region instead to be $(0,c\pi \tau]$, where $c$ is strictly between 1 and 2. This can lead to some improvements depending on the choice of $\eta$, which we will not discuss here for simplicity.
\end{remark}

More explicit bounds are obtained by choosing specific examples for the function $\eta$. For instance, we have the following elementary lemma:

\begin{lemma}\label{l6}
If a continuous function $\eta: \RR_- \to \RR_+$ satisfies $\eta(-k) = O(k^n)$ as $k \to \infty$ for some fixed $n>1$, then $K_\eta(\tau) = O(\tau^{-n+1})$, or if $\eta(-k) = O(e^{k^\alpha})$ as $k \to \infty$ for some fixed $0<\alpha \le 1$, then $K_\eta(\tau) = O(\tau^{1-\alpha}e^{-(\pi\tau)^\alpha})$ as $\tau \to \infty$. 
\end{lemma}

\begin{proof}
We give a proof of this lemma in the second case. Our conventions for the big-$O$-notation 
mean that $(1+e^{-k+u})\eta(-k) \ge C(1+e^{-k+u})e^{k^\alpha}$ for some constant $C>0$. Consider first the 
case that $0 \le k \le (2\alpha)^{1/(1-\alpha)}$. Then we get $(1+e^{-k+u})\eta(-k) \ge C\exp(-(2\alpha)^{1/(1-\alpha)}) e^u \ge O(e^{u^\alpha})$
when $u \to \infty$. In the other case when $k > (2\alpha)^{1/(1-\alpha)}$, the infimum of $k \mapsto (1+e^{-k+u}) e^{k^\alpha}$ is either attained 
for $k=(2\alpha)^{1/(1-\alpha)}$ -- which we have already discussed -- or at a stationary point $k_0$. 
Computing the derivative of this function  and using the condition $k > (2\alpha)^{1/(1-\alpha)}$, we find that at
the stationary point, we must have $u \le k_0$. But then $(1+e^{-k+u})\eta(-k) \ge C(1+e^{-k_0+u})e^{k^\alpha_0} \ge O(e^{u^\alpha})$, again, 
when $u \to \infty$. These two cases imply that $[\inf_{k \ge 0}(1+e^{-k+u})\eta(-k)]^{-1} \le O(e^{-u^\alpha})$. Thus, there exists a constant $c$
such that 
\ben
K_\eta(\tau) \le c  \int_{\pi \tau}^\infty e^{-u^\alpha} \dd u = \frac{c}{\alpha} \Gamma\left(\frac{1}{\alpha}, (\pi \tau)^{\alpha} \right) =
O(\tau^{1-\alpha}e^{-(\pi \tau)^\alpha})
\een
as $\tau \to \infty$. Here $\Gamma(p,y)$ is the incomplete Gamma function. The other case is treated similarly. 
\end{proof}

Combining theorem \ref{th2} with this lemma, we immediately get:

\begin{proposition}\label{prop1}
Let $U \in \gM_2$ be a unitary such that $\sigma_1^t(U) \in \gM_2$ for $|t| \le \tau$. 
\begin{enumerate}
\item For fixed $n>1$ we have that
\ben
\begin{split}
& | \langle U\Omega | (\log \Delta_1) U \Omega \rangle - \langle U\Omega | (\log \Delta_2) U \Omega \rangle | \\
\le & O(\tau^{-n+1}) \langle U \Omega |  (1-E_1^-\log \Delta_1)^n U \Omega \rangle , 
\end{split}
\een
for large $\tau$ uniformly in $U$. 
\item For fixed $0<\alpha\le 1$ we have that
\ben
\begin{split}
&| \langle U\Omega | (\log \Delta_1) U \Omega \rangle - \langle U\Omega | (\log \Delta_2) U \Omega \rangle | \\
\le & O(\tau^{1-\alpha}e^{-(\pi \tau)^\alpha}) \, \langle U \Omega |  \exp[(-E_1^-\log \Delta_1)^\alpha] U \Omega \rangle , 
\end{split}
\een
for large $\tau$ uniformly in $U$. 
\end{enumerate}
\end{proposition}

We remark that if the assumption of the proposition is satisfied for $U$, then also for $U^*$. In sec. \ref{sect1}, we apply the proposition to $U^*$.

\section{Applications to quantum field theory}
\label{sect3}

We now apply the abstract result Thm. \ref{thm1} in the context of quantum field theory (QFT). 
In the algebraic formulation of QFT, the algebraic relations between the quantum fields are encoded in a collection of $C^*$- or v. Neumann algebras associated with spacetime regions. The precise framework depends somewhat on the type of theory, spacetime background etc. one would like to consider.

In the case of Minkowski space $\bM=\RR^{d,1}$, a standard set of assumptions, manifestly satisfied by many examples, and believed to be satisfied by all reasonable QFTs, is as follows. Call a ``causal diamond'' $O \subset \bM$ any set of the form $O = D(A)$, where $A$ is any open subset of a Cauchy surface $\cong \RR^d$, and $D(A)$ its domain of dependence, i.e. the set of points $x \in \bM$ such that any inextendible causal curve through $x$ must hit $A$ once, see \cite{wald_3} for further details on these concepts. This is illustrated in fig.~\ref{fig:diamond}. 

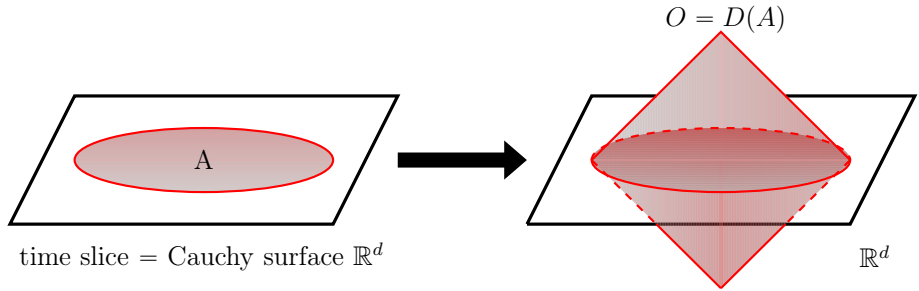
\begin{figure}[h!]
\begin{center}
\begin{tikzpicture}[scale=.85, transform shape]
	\draw[color=black, very thick] (-4, -1) -- (1, -1) -- (2, 1) -- (-3, 1) -- cycle;
	\fill[
	      top color=red!60,
	      bottom color=red!10,
	      shading=axis,
	      opacity=0.25
	      ]
	    (-1,0) circle (2cm and 0.5cm);
	    \draw[red, thick] (-3, 0) arc (180: 360: 2cm and 0.5cm) -- (1, 0);
	    \draw[red, thick] (-3, 0) arc (180: 0: 2cm and 0.5cm);
	
	\draw[line width=2mm,>={Triangle[length=3mm,width=5mm]},->] (2, 0) -- (4, 0);
	
	\draw[color=black, very thick] (4, -1) -- (9, -1) -- (10, 1) -- (8, 1);
	\draw[color=black, very thick] (4, -1) -- (5, 1) -- (6, 1);
	\fill[
	      left color=red!0,
	      right color=red!60,
	      middle color=red!30,
	      shading=axis,
	      opacity=0.25
	    ]
	    (5, 0) -- (7, 2) -- (9, 0) arc (0: 0: 2cm and 0.5cm);
	\fill[
	      left color=red!10,
	      right color=red!30,
	      middle color=red!60,
	      shading=axis,
	      opacity=0.25
	    ]
	    (5, 0) -- (7, -2) -- (9, 0) arc (0: 0: 2cm and 0.5cm);
	\fill[
	      top color=red!60,
	      bottom color=red!10,
	      shading=axis,
	      opacity=0.5
	      ]
	    (7,0) circle (2cm and 0.5cm);
	\draw[red, thick] (5, 0) arc (180: 360: 2cm and 0.5cm) -- (7, 2) -- cycle;
	\draw[color=red, dashed, thick] (5, 0) arc (180: 0: 2cm and 0.5cm);
	\draw[color=red, dashed, thick] (5, 0) -- (6, -1);
	\draw[color=red, dashed, thick] (9, 0) -- (8, -1);
	\draw[color=red, thick] (6, -1) -- (7, -2);
	\draw[color=red, thick] (8, -1) -- (7, -2);
	
	\node at (-1, 0) {A};
	\node[right] at (-4, -1.5) {time slice $=$ Cauchy surface $\RR^d$};
	\node[right] at (6, 2.2) {$O=D(A)$};
	\node[right] at (9, -1.5) {$ \RR^d$};
    \end{tikzpicture}

    \end{center}
    \caption{Causal diamond associated with $A$.}
\label{fig:diamond}
    \end{figure}

Poincar\' e transformations $g=(\Lambda, a) \in {\rm SO}_0(d,1) \ltimes \RR^{d+1}$ act on points in $\bM$ by $g\cdot x = \Lambda x + a$. Since Poincar\' e transformations are isometries of Minkowski spacetime, they map causal diamonds to causal diamonds, so we get  an action $O \mapsto g \cdot O$ on the set of causal diamonds.

Abstractly, a QFT  can be thought of as a collection (``net'') of $C^*$-algebras $\A(O)$ subject to the following conditions \cite{haag_1,haag_2}:

\begin{enumerate}
\item[a1)] (Isotony) $\A(O_1) \subset \A(O_2)$ if $O_1 \subset O_2$. We write $\A = \overline{\bigcup_O \A(O)}$ with completion in the $C^*$-norm.
\item[a2)] (Causality) $[\A(O_1),\A(O_2)]=\{0\}$ if $O_1$ is space-like related to $O_2$. In other words, algebras for space-like related double cones commute. Denoting the causal complement of a set $O$ by $O'$, we may also write this more suggestively as
$$
\A(O') \subset \A(O)'
$$
where the prime on the right side is the commutant.
\item[a3)] (Relativistic covariance) For each Poincare transformation $g \in \cP = {\rm Spin}_0(d,1) \ltimes \RR^{d+1}$ covering\footnote{The covering group is needed to describe non-integer spin.} a Poincar\'e transformation $(\Lambda,a) \in  {\rm SO}_0(d,1) \ltimes \RR^{d+1}$, there is an automorphism $\alpha_g$ on $\A$ such that $\alpha_g \A(O) = \A(g \cdot O)$ for all causal diamonds $O$ and such that $\alpha_g \alpha_{g'} = \alpha_{gg'}$ and $\alpha_{(1,0)}=\id$ is the identity.
\item[a4)] (Vacuum) There is a unique state $\omega_0$ on $\A$ invariant under $\alpha_g$. On its GNS-representation $(\pi_0, \H_0, |0\rangle)$, $\alpha_g$ is implemented by a projective positive energy representation $U$ of $\cP$ in the sense that $\pi_0(\alpha_g(a)) = U(g) \pi_0(a)U(g)^*$ for all $a \in \A, g \in \cP$. Positive energy means that the representation is strongly continuous, and that, if $x \in \bM \subset \cP$ is a translation by $x$, so that we can write
\ben
U(x) = \exp(-i P^\mu x_\mu),
\een
the vector generator $P=(P^\mu)$ has spectral values $p=(p^\mu)$ in the forward lightcone
 $p \in \bar{V}^+ = \{ k \in \RR^{d,1} \mid -(k^0)^2 + (k^1)^2 + \dots + (k^d)^2 \le 0, k^0>0\}$.
 \item[a5)] (Additivity) We assume that any element of $\A$ can be approximated arbitrarily well 
in the sense of matrix elements in the vacuum representation $\pi_0$ by finite sums of elements of the form $\alpha_{x_i}(a_i)$, where $a_i$ 
are in some arbitrarily small double cone, and where $\alpha_{x_i}$ denotes a translation by $x_i \in \bM$.
\end{enumerate}

For technical reasons, one often forms the weak closures $\gM(O) = [\pi(\A(O))]''$ of representations $\pi$ of the observable algebras. 
The double prime means the twice repeated commutant which if we start with 
a v. Neumann algebra would give back the algebra itself, and otherwise gives the smallest v. Neumann algebra containing the algebra we started with.
This gives a, in general representation dependent, net of v. Neumann algebras (on the respective representation Hilbert space $\H$). 

A straightforward, but important, consequence of axioms a1)-a5)  is the Reeh-Schlieder theorem \cite{reeh}, which is the following. We know by construction that $\pi_0(\A) |0 \rangle$ is dense in the entire Hilbert space, $\H_0$. One might guess at first that the subspace of states $\pi_0(\A(O)) |0 \rangle$, describing excitations relative to the vacuum localized in a double cone $O$, would depend on $O$. This expectation is  incorrect, however, and instead the Reeh-Schlieder theorem holds:
For any double cone $O$, the set of vectors $\pi_0(\A(O)) |0\rangle$ is dense in the entire Hilbert space. The same statement remains true if $|0\rangle$ is 
replaced with a vector with finite energy or by a KMS-state. 

The Reeh-Schlieder theorem implies that the vacuum vector $|0\rangle$ in the vacuum representation, or the vector representative $|0_\beta\rangle$ of a KMS state in a thermal representation is cyclic and separating for any double cone, so we are naturally in the setting of sec. \ref{sect1} and thm. \ref{thm1}. We now discuss these examples. 

\subsection{Touching regions in vacuum}

Consider first the following geometric situation: $\textcolor{red}{A_2} \subset \RR^d$ is some region in a spatial slice $\RR^d$ having $x^0=0$, 
$\textcolor{green}{A_1} \subset \RR^d$ is a half-plane in the same slice, e.g. $A_1= \{x^0=0, x^1 > 0 \}$. It is assumed that $A_2 \subset A_1$ and that both regions touch
at one boundary point, taken to be $0$ without loss of generality.
$O_j = D(A_j), j=1,2$ are the corresponding causal diamonds. 
We choose the vacuum representation $\pi_0$ (see a4) of the net, and set 
\ben
\label{Mdef}
\gM_j = \pi_0(\A(O_j))'',  \quad \H = \H_0, \quad |\Omega\rangle = |0\rangle, \quad \omega = \omega_0. 
\een  
Consider now a third region $B \subset A_2$ which is a ball of diameter $1$ centered at $(\frac12,0,0 \dots,0)$ (here we mean a point in a spatial slice $\RR^d$ having $x^0=0$). Then $\ell B$ is a region inside $A_2$ tangent to the point $0$ where the regions $A_1$ and $A_2$ touch each other and at the same time shrinking to zero size as $\ell \to 0$, see fig. \ref{fig:concentric}. 

\begin{figure}[h!]
\begin{center}

 \begin{tikzpicture}[scale=.65, transform shape]
        \filldraw[color=green!30, fill=green!10, very thick] (7, 4) -- (7, -4) -- (-1, -4) -- (-1, 4)-- cycle;
        \filldraw[color=green!10, fill=white, very thick] (0, 4) -- (0, -4) -- (-1, -4) -- (-1, 4)-- cycle;
        \draw[line width=1pt, color=green!60] (0,-4) -- (0,4);
	\filldraw[color=red, fill=red!10, very thick](3,0) circle (3);
	\filldraw[color=gray!60, fill=gray!10, very thick](.5,0) circle (.5);
	
	\draw[line width=1pt, >={Triangle[length=1mm,width=2mm]}, ->] (0,-.6) -- (1, -.6);
	\draw[line width=1pt, >={Triangle[length=1mm,width=2mm]}, ->] (1,-.6) -- (0, -.6);
	\filldraw (-.0,0) circle (2pt);
	
	\node[below] at (4.2, 0) {$A_2$};
	\node[above] at (.5, -0.3) {\textcolor{gray}{$\ell B$}};
	\node[above] at (6.5, 1) {$A_1$};
	\node[below] at (.5, -0.7) {$ \ell $};
	\node[left] at (-.1, 0) {$0$};
    \end{tikzpicture}

   \end{center}
    \caption{The regions $A_j, \ell B$.}
\label{fig:concentric}
    \end{figure}
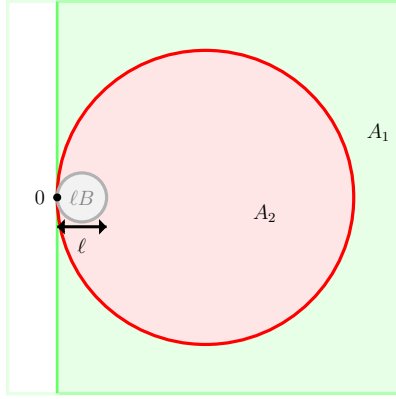

Consider furthermore a sequence of unitaries $U_\ell$ each of which is contained in the algebra associated with the double cone $D(\ell B)$, see see fig. \ref{fig:concentric}. In order to apply thm. \ref{thm1}, we need to know the maximum ``Rindler time'' value  $\tau$ such that $\sigma_1^t(U_\ell)$, the modular flow of the wedge, stays within the cone algebra $\gM_2$ for all $|t|\le \tau$. Because the modular flow of the wedge acts geometrically  by a 1-parameter family of boosts in the $(x^0,x^1)$-plane by the Bisognano-Wichmann theorem \cite{bisognano}, 
\ben
\label{boosts}
\sigma_1^t = \alpha_{\Lambda(t)}, \quad \Lambda(t) = 
\left(
\begin{matrix}
\sinh 2\pi t & \cosh 2\pi t & 0 & \dots & 0 \\
\cosh 2\pi t & \sinh 2\pi t & 0 & \dots & 0 \\
0 & 0 & 1 & \dots & 0 \\
\vdots &  &  &  & \vdots \\
0 & 0 & 0 & \dots  & 1
\end{matrix}
\right),
\een
the answer can be found without difficulty. It is exactly $ \tau = (2\pi)^{-1} |\log \ell/2R|$ if $A_2$ is a ball of radius $R$ touching the half space $A_1$
in the origin, see fig. \ref{fig:concentric}. 

So if we write 
\ben
\omega_{\ell}(a) \equiv \omega(U^*_\ell a U_\ell^{}) 
\een 
for the state excited by $U_\ell$
(represented by the vector $U_\ell |0\rangle$) and $\omega$ for the vacuum state (represented by the vector $|0\rangle$), 
case 2) of thm. \ref{thm1} gives in this case for instance
\ben
\label{Sineq}
| S_1(\omega / \omega_{\ell}) -  S_2(\omega / \omega_{\ell}) | 
\le  O\bigg(
 (\log \sqrt{R/\ell})^{1-\alpha} \exp \left[-  (\log \sqrt{R/\ell})^\alpha \right] \bigg)
 \quad \text{ as $\ell \to 0^+$,}
\een
 assuming in this case that our unitaries $U_\ell$ have been chosen e.g. such that 
\ben
\label{Mcond}
\| e^{|M|^\alpha} U_\ell^* |0\rangle \| \le C \quad \text{as $\ell \to 0^+$.}
\een 
Here $M = -i\frac{\dd}{\dd t} U(\Lambda(t))_{t=0}$ is the generator of the boosts in the $(x^0,x^1)$-plane given by a4) and \eqref{boosts}, so  
$e^{it M} = \Delta_1^{it}$ by the Bisognano-Wichmann theorem. 

Such a choice is possible 
generically if $0<\alpha < 1$ (but not for $\alpha=1$). In a dilation invariant theory, this will follow if we can chose {\em one} unitary, $U_1$, in 
$B$ such that the condition is satisfied, by simply setting $U_\ell = e^{i(\log \ell) D} U_1 e^{-i(\log \ell) D}$ for arbitrary $1 \ge \ell > 0$. 
Here $D$ is the generator of dilations on $\H$. This follows because $M$  commutes with $D$. In order to suggest that for any $0<\alpha < 1$, but not $\alpha\ge 1$, there typically ought to exist a unitary, $U_1$, in $B$ such that  $U^*_1 |0 \rangle$ is in the domain of $e^{|M|^\alpha}$, we consider below as an illustrative example of a chiral half of the free massless fermion field in 1+1 dimensions.

However, before that, we point out that we can immediately generalize the above result to more general pairs of open regions $A_1, A_2 \subset \RR^d$ as in fig. \ref{fig:regions}:

\begin{theorem}
\label{ttouch}
Let $A_1 \supset A_2$ be convex, open regions in $\RR^d$ touching in a single point $p$ on their boundaries. Assume furthermore that there exists an 
open ball of radius $R$ contained in $A_2$ whose boundary touches $p$. Let $O_j = D(A_j), j=1,2$ be the causal completions and $\gM_j$ the corresponding 
algebras of observables as in \eqref{Mdef}. Let $\omega$ be the vacuum state of a theory satisfying a1)--a5). Then if $\{U_\ell\}_{\ell > 0}$ is a family of unitary operators as described, satisfying \eqref{Mcond} for the generator of boosts $M$ in the half-space containing $A_1$, and touching $p$, whose spacetime localization shrinks to $p$ as $\ell \to 0^+$, then \eqref{Sineq} holds in this limit, where $\omega_\ell( \ . \ )=
\omega(U_\ell^* \ . \ U_\ell)$ is the state excited by $U_\ell$. 
\end{theorem} 

\begin{remark}
With the improved bound indicated in remark \ref{improv}, the upper bound \eqref{Sineq} can easily be improved e.g. to the bound \eqref{Sineq1} mentioned in the introduction. As the example of the free massless fermion in the next section suggests, the value of $\alpha$ must be $<1$, so the decay in \eqref{Sineq} falls 
short of the limiting behavior $O(\sqrt{\ell/R})$. 
\end{remark}

\begin{proof}
Let $A_3$ be a half-plane whose boundary touches $p$ and such that $A_3 \supset A_1, A_2$ (which exists due to convexity), and let $A_4$ be an
open ball of radius $R$ contained in $A_2$ whose boundary touches $p$ (which exists by assumption). By the monotonicity of the relative entropies and 
a1), we have, with the obvious notations, 
$S_3 \ge S_1 \ge S_2 \ge S_4$, implying $|S_1-S_2| \le |S_3-S_4|$. However, we have already argued that the claimed bound \eqref{Sineq}
holds  for $|S_3-S_4|$, which finishes the proof.   
\end{proof}

\subsection{Free massless fermions in $1+1$ dimensions}

In order to illustrate the meaning of the condition \eqref{Mcond} entering the assumption of thm. \ref{ttouch}, we now consider the theory of free massless 
fermions in 2-dimensional Minkowski space. As is well known, such a theory can be viewed as the tensor product of two ``chiral halves'', each living on a lightray. 
The discussion boils down to the discussion of these theories on the light ray, and so, for simplicity, we directly focus on them. 

For conformal field theory on one lightray, the axioms a1)-a5) are formulated in a somewhat adapted form, so we first describe this. Instead of a1), we now have a
net $\{ \gA(I) \}$ indexed by open intervals $I \subset \RR$, with $\RR$ thought of as representing one light ray. a2) remains unchanged except that the notion of 
complement is now the ordinary complement of subsets of $\RR$. In a3), we have instead of the Poincare group now the group of 
conformal maps of the ightray, isomorphic to the M\" obuis group ${\rm PSL}_2(\RR)$, where a group element $g=\left(
\begin{matrix}
a & b\\
c & d
\end{matrix}
\right)$ acts by $g(x) = \frac{ax+b}{dx+c}$. In a4) we now have a projective unitary positive energy representation of the M\" obius group, where $P$ is now the generator of translations $x \mapsto x+b$. a5) remains the same. 

The algebras for one chiral half of the free massless Fermion theory are described as follows, see \cite{araki_5,dantoni} for details. The algebra $\gA(I), I \subset \RR$ an open interval is generated as a $C^*$-algebra by the symbols $\psi(f)$, where $f \subset C_0^\infty(I, \CC)$ is a testfunction supported in $I$, and the identity $1$, subject to the CAR relations: $f \mapsto \psi(f)$ linear, $\psi(f) \psi(h) + \psi(h) \psi(f) = (\Gamma f, h) 1$, $\psi(f)^* = \psi(\Gamma f)$, with $( \ , \ )$ the inner product in $L^2(\RR)$ and 
with $\Gamma f(x) = \overline{f(x)}$. The unique $C^*$-norm compatible with these relations is described in \cite{araki_5}, here we only need to know that 
$\| \psi(f) \|^2 = \frac12 (f,f)$ when $f$ is real-valued. Then the relations imply that $\psi(f)=\psi(f)^* = \psi(f)^{-1}$ is unitary when $(f,f) = 2$ and when $f$ is real-valued, which follows immediately from the relations and the properties of the $C^*$-norm. The unique vacuum state satisfying a4) is the unique Gaussian (``quasi-free'' in the terminology of \cite{araki_5}) state specified by the 2-point function
\ben\label{twopoint}
\omega(\psi(h) \psi(f) ) = \frac{i}{2\pi} \int  \frac{h(x) f(y)}{x-y-i0} \, \dd x \dd y . 
\een
Informally, we think of $\psi(f)=\int \psi(x) f(x) \dd x$ as a smeared version of the local (singular)
quantum field $\psi(x)$. We do not describe explicitly the corresponding vacuum representation $\pi_0$, as we will not need its explicit form. It is built on a fermionic Fock-space with vacuum vector $|0\rangle$ representing the above state functional $\omega$.

In order to make contact with the setting described in the previous section, we now set
$\gM_1 = \pi_0(\gA((0, \infty)))'', \gM_2 = \pi_0(\gA((0,1)))''$. One has an analogue of the Bisognano-Wichmann theorem \cite{Hislop:1981uh}, which implies that 
the modular flow $\sigma^t_1$ of $\gM_1$ is geometrically described by the dilations, i.e. the M\"obuis group elements $g(t)  =\left(
\begin{matrix}
e^{\pi t} & 0\\
0 & e^{-\pi t}
\end{matrix}
\right)$, acting on a point by $g_t(x) = e^{2\pi t}x$. In other words, $\Delta_1^{it} = e^{itD}$, where $D=-i\frac{\dd}{\dd t} U(g(t))_{t=0}$ is the 
rescaled generator of dilations in the vacuum representation of the M\"obuis group, a4). We now let $f$ be a real-valued, smooth test-function
supported in $(0,\frac12)$ with $\int f(x)^2 \dd x = 2$, and we define, for $\ell>0$
\ben\label{ULdef}
U_\ell := \psi(f_\ell), \quad f_\ell(x) = f(\ell^{-1}x)/\sqrt{\ell}
\een
It follows that each $U_\ell$ is a unitary operator contained in the local algebra associated with the interval $(0,\frac12 \ell)$. The analogue of thm. \ref{ttouch} for the case at hand is that \eqref{Sineq} (with $R=1$) holds for $\ell \to 0$ provided \eqref{Mcond} (with $M$ replaced by $D$) is satisfied. We would now like to see what it means for $f$ to be such that the condition \eqref{Mcond} is indeed satisfied. Using the spectral theorem we see that this is equivalent to 
\ben
C^2 \ge \| e^{|D|^\alpha} U_\ell^* |0\rangle \|^2 = \int_\RR 
e^{2|s|^\alpha} \left( \int_\RR e^{its} \langle 0 | \psi(f_\ell) e^{itD} \psi(f_\ell) | 0 \rangle \frac{\dd t}{2\pi} \right) 
\dd s 
\een
uniformly in $\ell$, where here and in the following, we identify $\psi(f)$ with their representatives 
$\pi_0(\psi(f))$ on the vacuum Hilbert space in a4). Next, we use $e^{itD} \psi(f_\ell) | 0 \rangle = \psi(f_{\exp(2\pi t)\ell}) | 0 \rangle$
from a3), a4), we use \eqref{twopoint}, and we define $h(u) = e^{\pi u} f(e^{2\pi u})$, which is another smooth test function of compact support. Using also 
\eqref{fourier}, this gives
\ben
\| e^{|D|^\alpha} U_\ell^* |0\rangle \|^2 = \int_\RR \frac{|\hat h(s)|^2}{1+e^s} e^{2|s|^\alpha} \, \dd s 
\een
after a short calculation for all $\ell>0$. The integral on the right converges for large $|s|$ if the decay of $|\hat h(s)|$ is  
$O( |s|^{-1-\epsilon} e^{-|s|^\alpha} ), \epsilon>0$, for example. It is possible to achieve this behavior provided $\alpha<1$ 
\cite{ingham}, but not for $\alpha=1$, as the latter would imply analyticity of $h(u)$, which would be in contradiction with the compact support property. Thus, we conclude:
\begin{proposition}
Let $f$ be a real valued test function supported in $(\frac12,1)$ with $\int f(x)^2 \dd x = 2$ such that the Fourier transform of 
$u \mapsto e^{\pi u} f(e^{2\pi u})$ is of order $O( |s|^{-1-\epsilon} e^{-|s|^\alpha} )$ for $|s| \to \infty, \epsilon >0$
(such functions exist iff $\alpha<1$). Then if $U_\ell$ are the unitaries defined in \eqref{ULdef}, and if $A_1=(0, \infty), A_2=(0,1)$, 
we get \eqref{Sineq} with $R=1$ for the free massless Fermi field on the lightray.  
\end{proposition}

\subsection{Large regions in thermal states}

Next we consider a thermal representation $\pi_\beta$ on a Hilbert space $\H_\beta$ with state vector $|0_\beta \rangle$ satisfying the KMS condition at temperature $\beta>0$. We let $O_1 = \mathbb M$ be the entire Minkowski space and $O_2$ a double cone of a ball of radius $r$ centered at the origin. 
We are going to let $r$ become large. The observable algebras (systems) are chosen to be:
\ben
\gM_j = \pi_\beta(\A(O_j))'',  \quad \H = \H_\beta, \quad |\Omega\rangle = |0_\beta \rangle, \quad \omega = \omega_\beta. 
\een 
It follows that the modular flow of $\gM_1$ is given by backward 
time-translations, i.e. $\sigma_1^t = \alpha_{-\beta t e}$, with $\alpha_{te}$ standing for the time-translation
automorphism (see a3)) into a  time-like direction fixed by a unit vector $e$ (the rest frame of the thermal bath). The commutant $\gM_1'$ is often called in this context the ``thermo-field double''. For finite dimensional systems as in the example in sec. \ref{sect1}, we would have  $\gM_1 = M_n(\CC) \otimes 1_n$, and
$|0_\beta \rangle = Z_\beta^{-1/2} \sum_{j=1}^n e^{-\beta E_j/2} |j\rangle \otimes |j \rangle$, where $E_j$ are the energy eigenvalues of some self-adjoint Hamiltonian $H$ on $\CC^n$. In this case, $\beta^{-1}\log \Delta_1 = -H \otimes 1 + 1 \otimes H \equiv -H_\beta$. This operator is sometimes called the ``Liouvillean''. 

Now let $U$ be a unitary in some fixed double cone of unit size centered about the origin and let $\omega_U = \omega( U^* . U)$ be the excited state
(represented by the vector $U |0_\beta\rangle$). Hence, the maximum time $\tau$ such that 
$\sigma_1^t(U)$ remains in $\gM_2$ for all $|t| \le \tau$ is of order $\tau \sim r/\beta$ for $r \to \infty$. Case 2) of thm. \ref{thm1} gives for instance
\ben
| S_1(\omega / \omega_{U}) -  S_2(\omega / \omega_{U}) | 
\le  O\bigg( (r/\beta)^{1-\alpha} e^{-(\pi r/\beta)^\alpha} \bigg),
\een
assuming in this case that our unitary $U$ is chosen such that $U^* |0_\beta \rangle$ is in the domain of 
$e^{|H_\beta|^\alpha/2}$, where $H_\beta$ is the generator of time-translation in the thermal representation\footnote{
According to thm. \ref{thm1}, a sufficient condition would be that $U^* |0_\beta \rangle$ is in the domain of 
$e^{(H_\beta^+)^\alpha/2}$, where $H_\beta^+$ denotes the part of $H_\beta$ that is projected onto the {\em positive} spectral subspace. The relative minus sign is due to the fact that $\log \Delta_1 = -\beta H_\beta$.}. On the other hand, if we merely know that $U^* |0_\beta \rangle$ is in the domain of $|H_\beta|^n$ for some $n > 1$, then we learn from case 1)  of thm. \ref{thm1} that 
\ben
| S_1(\omega / \omega_{U}) -  S_2(\omega / \omega_{U}) | 
\le  O((r/\beta)^{-n+1})
\een
which is evidently a weaker decay. These examples should suffice to illustrate how to apply thm. \ref{thm1}.

\section{Conclusions}

Our main result \eqref{Sineq1} can be stated in a less precise fashion as saying that, if $\rho_A$ is the reduced density matrix of the vacuum state for a region $A$, then $S(\rho_A / U \rho_A U^*)$ is independent of the global shape of $A$ when the localization of a unitary $U$ converges to a point $p$ on the boundary $\partial A$. It is perhaps possible to say more along the following lines. One can look at $-S(\rho_A / U \rho_A U^*)$ in the spirit of the 1st law of thermodynamics as $\Delta S(p) - T(p) \Delta E(p)$ \cite{Casini:2008cr} (see also \cite{Longo:2018zib}). 
Here $\Delta S(p)=S_{\rm vN}(\rho_A)- S_{\rm vN}(U \rho_A U^*)$ is the difference between the v. Neumann entropies and $T(p)$ should be  
a local temperature in the spirit of \cite{Arias:2016nip}, in the limit when the localization of $U$ approaches $p$. 
The idea would then be that $T(p)$ only depends on the geometry of $\partial A$ at $p$. It would be interesting to investigate this further in a general setting, perhaps along the lines of \cite{Arias:2016nip}. 

\medskip
\noindent
{\bf Acknowledgements:} It is a pleasure to thank Centro Atomico Balseiro, Bariloche, Argentina, 
for hospitality during my visit in March 2018, as well as the Simons Foundation 
for financially supporting that visit. I have greatly benefited from discussions with H. Casini, M. Huerta, and 
D. Pontello.

\end{document}